\documentclass{article}
\usepackage[utf8]{inputenc}
\usepackage{amssymb, amsmath, verbatim, amsthm,url, multirow,fullpage,mathtools, appendix}
\usepackage{longtable, rotating,makecell,array}
\usepackage[aligntableaux=top]{ytableau}

\usepackage{soul}
\usepackage[colorinlistoftodos,textsize=footnotesize]{todonotes}

\setstcolor{red}
\usepackage[all]{xy}
\xyoption{poly}
\xyoption{arc}
\def\ba #1\ea{\begin{align} #1 \end{align}}
\def\bas #1\eas{\begin{align*} #1 \end{align*}}
\def\bml #1\eml{\begin{multline} #1 \end{multline}}
\def\bmls #1\emls{\begin{multline*} #1 \end{multline*}}
\newtheorem{thm}{Theorem}[section]

\newtheorem{lem}[thm]{Lemma}

\newtheorem{algorithm}[thm]{Algorithm}

\theoremstyle{remark}
\newtheorem{eg}[thm]{Example}

\theoremstyle{definition}
\newtheorem{dfn}[thm]{Definition}

\title{Isolated loops}
\author{Caroline Mosko}
\date{}

\begin{document}

\maketitle

\begin{abstract}
    Many bureaucratic and industrial processes involve decision points where an object can be sent to a variety of different stations based on certain preconditions. Consider for example a visa application that has needs to be checked at various stages, and move to different stations based on the outcomes of said checks. While the individual decision points in these processes are well defined, in a complicated system, it is hard to understand the redundancies that can be introduced globally by composing a number of these decisions locally. In this paper, we model these processes as Eulerian paths and give an algorithm for calculating a measure of these redundancies, called isolated loops, as a type of loop count on Eulerian paths, and give a bound on this quantity.
\end{abstract}

Many processes involve a person or a file or a manufactured item traveling from station to station, or traveling through a series of states in order to be completed. Some of these may be simple processes with a few clearly defined steps, such as the steps needed to apply for a drivers license. Most people follow the same steps in the same sequence, but certain classes of people have to go through alternate processes (for example, first time applicants or older drivers, or drivers with past poor records). Other processes are inherently complicated, with many stages that need to be completed, and with many points where decisions have to be made as to what the next step in the process for a particular case should be. For instance, a manufactured item may undergo multiple quality control steps, where it is pulled aside for irregularities, processed separately, then returned to an earlier stage in the processing pipeline to address any safety concerns.  Similarly, bureaucratic processes, such a visa application,  have many decision points where the authorities may wish to investigate personal details or request additional information. At each of these decision points, there are multiple possible next steps the application may need to undergo, some of which may lead to the application revisiting a step that it has already visited.

More generally, we consider processes where a case passes through a series of stations. At each station, depending the specifics of the case, it is processed and then a decision is made as to where it is sent next. Each case can only occupy one stage at a time. We assume that at each stage, the rules for the decision are well defined and consistent. However, taken as a whole, the sets of choices at each stage can cause a case to return to a station it has already visited. In this paper, we call the phenomenon of returning to an earlier stage in the process a \emph{redundancy}.

We model this process graphically by representing the various stations as vertices, and the transitions of the case through the process as a edges. In other words, we represent the path of each case through the process as an Eulerian path in a multi-digraph (defined in section \ref{sec:processesandgraph}). The goal of this paper is to count the number of edge distinct loops experienced by a case in a process. Due to the overlapping nature of loops in graphs, counting the number of edge distinct loops is not straight forward. Algorithm \ref{alg:cara} gives a range of possible values for this count, while Theorem \ref{res:upperbound} gives an upper bound for these values. 

There is a long history of using graphical methods for studying processes where a single object passes from a single state to another state based on what preconditions it has met. For instance, Petri Nets \cite{Petrinets} are a powerful way to encapsulate not only where an a case could go, but also allows one to put additional constraints like time or cost to the process. Rules based processes like the ones studied in the paper are also often studied in terms of decision trees \cite{decisiontrees} which is a clear tool for understanding the rules at each station that determines the next station the case will visit. Some pieces of software can be thought of as rules bases systems, where a related measure, called cyclotomic complexity measures the complexity of the software \cite{cyclotomic}. If the rule set for deciding where the case will go next has a random component to it, but is purely a function of the current state, and not on any of the previous stations it has been to, the entire process is called Markovian \cite{Markovprocess}. There is an entire field of mathematics devoted to the study of such processes. However, tools such as these encapsulate the complexity of the process as a whole. They do not include a means of measuring how frequently a case may revisit an earlier stage in the processing stream, or how many times this may happen to a single case. 

In this paper, we do not attempt to judge a process as a whole. Rather, we give a measure of redundancy experienced by a single case, called \emph{isolated loop count} (see Definition \ref{def:isolated}) as a purely graph theoretic object. We propose that by aggregating this count for all cases passing through a system, one may get a better idea of the redundancies of the system, based on the specific outcomes that it wishes to achieve.

There are many existing methods of counting loops on graphs of various sorts, from the Euler's loop number of a graph to algorithms to identify and count oriented loops on graphs \cite{findingcycles, countingcycles} to algorithms that count the number of Eulerian circuits contained in a graph \cite{BEST}. However, to our knowledge, there is no algorithm to count the number of \emph{distinct} redundancies in a given Eulerian path. For the purposes of this paper, we say that two redundancies are isolated if they do not share a common edge. This reason for this definition is that we wish to avoid counting the same transition in the lifetime of a case undergoing a process multiple times when the case only experiences that transition once. 

Algorithm \ref{alg:cara} gives an algorithm for counting the non-overlapping circuits of an Eulerian path, which, according to Definition \ref{def:isolated}, we refer to as isolated loops in this paper. There are several subtleties in how one arrives at this quantity. Namely, a complicated process, represented as a graph, may contain many different loops. However, all of these loops may not correspond to valid sequences in the process (see Example \ref{eg:notpathrespecting}). Furthermore, for any given multi-digraph, there may be multiple different Eulerian paths in it \cite{BEST} (see Section \ref{sec:BEST}). We are only interested in counting the loops in a given path. Changing the path may give a different count. Unfortunately, there is not a unique count of this value for all graphs. In this paper we give an algorithm for calculating the range of values that this quantity can take, and conditions when this quantity is unique.

Finally, while we call the isolated loop count a measure of redundancy experienced by a case, we do not mean it in any pejorative sense. In the visa application example discussed above, it may well be possible that the application was pulled aside because it was missing certain pages and returned to the system after the appropriate materials were supplied. This may not be a negative part of the process. Or, it may be the case that the application was set aside automatically because it had been in the system for too long. Depending on the reasons for the delay, this may be a concerning inefficiency of the process. In our framework, we treat the loops that arise from both these scenarios the same.

In Section \ref{sec:comparison}, we compare the measure of redundancies developed in this paper to other algorithms for counting loops in multi-digraphs. 

\section{Eulerian paths and Isolated loops}
In this section, we introduce Eulerian paths and isolated loops. The Eulerian paths in this paper model processes, bureaucratic, industrial, or otherwise, where a person or an object has to pass through several stations in order to complete a process. One may think of this as a piece of machinery that has to go through various stages of an assembly line, with various safety checks along the way causing certain members of the population to experience a different series of steps than its peers. One may also think of the process underlying a visa application, or some other bureaucratic process in a similar vein.

While each of these processes may have a consistent set of rules governing where a particular piece of machinery or application is sent next, for large and complex systems, there may not be a global understanding of what emerges when these local decision rules are put together. The purpose of counting the isolated loops defined in this section is to get an understanding of any redundancies that may have been built into the system as a whole by putting together internally consistent transition rules at each decision point.

We begin by presenting the graph theoretical background needed to model these processes. 

\subsection{Eulerian Paths\label{sec:processesandgraph}}

A \emph{graph} $G$ is a given set of points (vertices) and lines (edges) that are connecting them. A \emph{multi-graph} is a graph where two vertices may have multiple edges between them. Two edges are adjacent if they share a vertex. The \emph{degree} of a vertex is the number of edges that are incident on said vertex.

\begin{dfn} A \emph{walk} on a multi-graph is a sequence of adjacent edges in a graph, sometimes represented as the sequence of vertices that that are the endpoints of the edges. A closed walk starts and ends at the same vertex. An open walk does not. \end{dfn}

All graphs we consider in this paper are connected. That is, there is an undirected walk from any vertex of the graph to any other.

In this paper, we say that a \emph{circuit} is a closed walk. A \emph{cycle}, or a loop, is a closed walk that does not contain any smaller closed walks as a subgraph.

A \emph{directed graph} is a graph where all the edges have an associated direction, i.e. they flow from a one vertex (called a source vertex) to another (called a target vertex). This is also called the \emph{orientation} of an edge. A \emph{multi-digraph} is a multi-graph with directed edges. Two directed edges are adjacent if the target vertex of one edge is the source vertex of the next. For a multi-digraph, one may talk about the \emph{in degree} and the \emph{out degree} of a vertex as the number of edges for which it is a target (or edges coming into the vertex) and the number of edges for which said vertex is a source (or edges coming out of the vertex). 

We say a walk on a multi-digraph respects the orientation of the graph it can be described as a sequence of adjacent edges, where adjacency is defined in terms of directed edges. I.e., if it passes from the source vertex of an edge to the target vertex of the same graph. A closed orientation respecting walk is an \emph{oriented circuit}, or an \emph{oriented cycle} if it does not contain any smaller closed walks as a subgraph. 

In this paper, we deal with multi-digraphs. In particular, we consider a special subset of multi-digraphs called Eulerian graphs. Theses are well studied objects that the interested reader can find more about in textbooks such as \cite{Graphtheorybook}. 

\begin{dfn}
An \emph{Eulerian path} on a graph is a walk that visits each edge once. The path is called open if the starting vertex is not the same at the final vertex, and a cycle if they are the same. An Eulerian path (resp. cycle) on a multi-digraph must further respect the orientations of the edges. A graph (resp. multi-digraph) is an \emph{Eulerian graph} if it contains an Eulerian path. \end{dfn} 

Note that this definition is well suited to model processes that are a sequence of transitions from one state to another. In particular, in this paper, we consider Eulerian graphs as a sequence of adjacent edges.  

\begin{lem} \label{res:edgeorder}
Given an Eulerian multi-digraph, an Eulerian path induces an ordering on edges and defines a sequence of vertices where each vertex appears at least once. \end{lem}
\begin{proof}
By the nature of an Eulerian path (passes through each edge exactly once) each edge can be assigned an ordering, corresponding to when it is traversed uniquely. 

As a walk passes through adjacent edges in a given order, the source of of an edge directly proceeds the target of said edge in the sequence. Furthermore, a vertex in the sequence directly follows another vertex in the sequence if and only if there is an edge in the graph such that the latter is the source vertex of the edge and the former the target. Since the graph is connected, every vertex appears at least once in the sequence.
\end{proof}

In this paper, our fundamental object is an Eulerian path, which we represent as a sequence of vertices in the order that they are visited. Graphically, we represent this as a multi-digraph. Therefore, even though several different Eulerian paths can give rise to the same graph, this is not a source of ambiguity in this setting. Note that in this way, the orientations of the edges of our multi-digraphs always agree with the walk we consider. 

\begin{eg} \label{eg:notpathrespecting}
Consider the digraph $G$ below. 
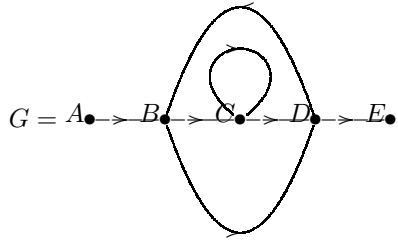
\bas G =   
{\xy
(-20, 0) *{\bullet} = "A" +(-2, 1)*{A},
(-10, 0) *{\bullet} = "B" +(-2, 1)*{B},
(0, 0) *{\bullet} = "C" +(-2, 1)*{C},
(10, 0) *{\bullet} = "D" +(-2, 1)*{D},
(20, 0) *{\bullet} = "E" +(-2, 1)*{E},
"A"; "B" **{\dir{-}}?/0pt/*{\dir{>}},
"B"; "C" **{\dir{-}}?/0pt/*{\dir{>}},
"C"; "D" **{\dir{-}}?/0pt/*{\dir{>}},
"D"; "E" **{\dir{-}}?/0pt/*{\dir{>}},
"D"; "B" **\crv{+(10,30)}?/0pt/*{\dir{>}},
"B"; "D" **\crv{+(-10,-30)}?/0pt/*{\dir{>}},
"C"; "C" **\crv{+(-7,5)&+(7,9)&+(7,-9)}?/0pt/*{\dir{>}},
\endxy}
\eas One can define multiple different orientations respecting Eulerian paths on this graph. Two examples are: \ba A \rightarrow B \rightarrow C \rightarrow C \rightarrow D \rightarrow B \rightarrow D \rightarrow E \;. \label{eq:path}\ea and \bas A\rightarrow B\rightarrow D\rightarrow B\rightarrow C\rightarrow C\rightarrow D \rightarrow E \;. \eas  We can also find non-orientation respecting Eulerian paths on this graphs: \bas A\rightarrow  B\rightarrow D\rightarrow C\rightarrow C\rightarrow B\rightarrow D\rightarrow E \;. \eas

However, in this paper, we start with an Eulerian path as a sequence of vertices, as in Lemma \ref{res:edgeorder} above. In particular, we write the path in \eqref{eq:path} as $ABCCDBDE$, and consider the graphs to be representations of the paths. Therefore, we avoid any ambiguity as to the ordering of the edges in the walk. 

\end{eg}

\begin{dfn}
Given an Eulerian path $P$, represented by a multi-digraph $G$, let $Q$ be a subset of $P$ such that two edges are consecutive in $Q$ only if they are consecutive in $P$. Then we say that $Q$ is a \emph{subpath} of $P$. We say that $Q$ is a \emph{subcircuit} of $P$ if $Q$ is also an oriented circuit in $G$.
\end{dfn}

In Example \ref{eg:notpathrespecting}, for $P = ABCCDBDE$, the ordered set of edges indicated by $Q = DBD$ is both a subpath and a subcircuit of $P$, while the sequence $R = DBCCD$ is not a subpath of $P$, even though it is an oriented circuit in G; the sequence $BCCD$ appears in $P$, and the edge $DB$ appears in $P$. However, while $DB$ proceeds $BCCD$ in $P$, it succeeds $BCCD$ in $R$.

\subsection{Redundancies and Isolated loops}

In this paper, the redundancies we are interested in are the occasions when a process doubles back on itself, i.e. when a vertex appears multiple times in the sequence of vertices defined by the Eulerian path in Lemma \ref{res:edgeorder}. Specifically, we are only interested in counting certain circuits of the graph  representing an Eulerian path, those defined in \ref{def:redundancy}.  
Note that this means we are not interested in every oriented cycle of a given Eulerian graph, but only those arising from a fixed Eulerian path. Neither are we interested only in subpaths of the given Eulerian paths. 

\begin{dfn}\label{def:redundancy}
A \emph{redundancy} of an Eulerian path $P$ is a sequence of edges, not necessarily a subpath, that forms an oriented cycle in $G$, such that if the target of edge $e_1$ is the source of edge $e_2$ in the oriented cycle, the edge $e_1$ proceeds $e_2$ in the path $P$. We say that a redundancy $\mathcal{C}$ is \emph{interrupted} if $\mathcal{C}$ is not a subpath of $P$.
\end{dfn}

We illustrate Definition \ref{def:redundancy} with an example.

\begin{eg} \label{eg:cycleoddities}
Consider the sequence \ba \begin{matrix} A &B &C &D &E &F &A &C &E &A \end{matrix} \;,\label{eq:subsequence}\ea  represented by the graph 
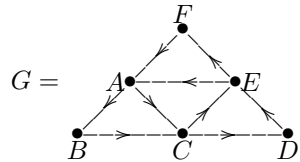
\bas 
G  = {\xy
(-14, -7) *{\bullet} = "B" +(0, -2)*{B},
(0, -7) *{\bullet} = "C" +(0, -2)*{C},
(14, -7) *{\bullet} = "D" +(0, -2)*{D},
(-7, 0) *{\bullet} = "A" +(-2, 0)*{A},
(7, 0) *{\bullet} = "E" +(2, 0)*{E},
(0, 7) *{\bullet} = "F" +(0, 2)*{F},
"A"; "B" **{\dir{-}}?/0pt/*{\dir{>}},
"B"; "C" **{\dir{-}}?/0pt/*{\dir{>}},
"C"; "D" **{\dir{-}}?/0pt/*{\dir{>}},
"D"; "E" **{\dir{-}}?/0pt/*{\dir{>}},
"E"; "F" **{\dir{-}}?/0pt/*{\dir{>}},
"F"; "A" **{\dir{-}}?/0pt/*{\dir{>}},
"A"; "C" **{\dir{-}}?/0pt/*{\dir{>}},
"C"; "E" **{\dir{-}}?/0pt/*{\dir{>}},
"E"; "A" **{\dir{-}}?/0pt/*{\dir{>}},
\endxy}
\eas

We can consider the circuit $Q = A B C D E F A C E A$ and its subcircuit $R = CDEFAC$, which are both subpaths of the original process. Note that if one removes $R$ from $Q$, then one is left with a cycle $Q \setminus R  = ABCEA$. While $Q \setminus R$ is not a subcircuit of $Q$, we wish to consider it as an interrupted redundancy of the system. We interpret this as the cycle $ABCEA$ in the modeled process having been interrupted at the point $C$ in order to perform the cycle $R$.

Notice also that the graph $G$ contains the oriented cycle $ABCEFA$. We do not consider this oriented cycle as a redundancy of our process, because the edge $CE$, which falls between the edges $BC$ and $EF$ in the oriented cycle, but after both those edges in the Eulerian path we wish to study. Therefore, the sequence $ABCEFA$ cannot be considered as a part of the Eulerian path.

\end{eg}

Note from Example \ref{eg:cycleoddities} that an interrupted redundancy may arise from a circuit being interrupted by another circuit, not just by another cycle. 

In this paper, we wish to count edge distinct redundancies. This is motivated by the idea of calculating the Euler loop number of a planar graph. If two unoriented cycles of a planar graph can be combined to form a third oriented cycle, we do not wish to count all three, as only two of the three are independent of each other. Similarly, redundancies in Eulerian paths are complicated and overlapping. We do not wish to count two redundancies from an overlapping set of edges, but only those that come from distinct portions of the process. 

Note that there are several other algorithms that count loops and oriented loops that are distinct from the algorithm developed in this paper, and not appropriate for our purposes (see section \ref{sec:comparison}). In particular,  we are not interested in counting the number of Eulerian cycles in a multi-digraph, which can be done by the BEST theorem \cite{BEST}. Nor do we want a straightforward calculation of the Euler characteristic. Instead we wish to count the number of non-overlapping portions of a given Eulerian path that form simple circuits. 

This is because, when we use Eulerian paths to represent the path of a case through a complicated process, we are interested in understanding when it returns to a point that it has already visited in the system. However, we wish to respect the transitions experienced by the case, i.e. the edges of the Eulerian path, as the fundamental object. Therefore we do not permit the same edge or transition occur multiple times among the redundancies we count. 

\begin{dfn} Two redundancies are called \emph{isolated} if they do not share an edge with each other, i.e. are edge distinct.\label{def:isolated}\end{dfn} 

In order to count isolated loops, we need to classify subcircuits of an Eulerian path into three categories:

\begin{dfn}
Given an Eulerian path $P$, two subcircuits $Q$ and $R$ are related in one of the following three ways.
\begin{enumerate}
\item The subcircuits $Q$ and $R$ are \emph{disjoint} if one circuit occurs entirely before the other. 
\item The subcircuits $Q$ and $R$ are \emph{nesting} if $R$ is a subcircuit of $Q$. If there is no other subcircuit of $P$, say $S$, such that $R$ is a subcircuit of $S$ which itself is a subcircuit of $Q$, then we say that $Q$ is the parent of $R$ or that $R$ is the daughter of $Q$.
\item The subcircuits $Q$ and $R$ are \emph{overlapping} if the corresponding intervals have a non-trivial intersection but one is not contained in the other. I.e. overlapping circuits are neither nested nor disjointed. 
\end{enumerate}
\end{dfn}

Note that the nesting relationship defines a poset, or partially ordered set, relationship on the circuits. Namely, a circuit may have multiple parent circuits. Note that if a circuit has more than one parent, the parents are, by construction, overlapping. This is because, if $P$ and $Q$ are both parents of $S$, then by construction, $S$ is contained in both $P$ and $Q$. For instance, consider the following Eulerian path: \bas \setcounter{MaxMatrixCols}{15} P = \begin{matrix} A & B & C & D & E & E & B & F & G & D &  H & \\  & [ &  &  &  &  & ] &  &  &  &  &  & (Q) \\  &  &  & [ &  &  &  &  &  & ] &  &  & (R) \\ &  &  & & [ & ] &  &  &  &  &  &  & (S)  \end{matrix} \;. \eas 
 
Note that if $Q$ and $R$ are overlapping subcircuits, $Q \cup R$ may not be a circuit, but a subpath of $P$ (it need not start and end at the same vertex). If there is a subcircuit $S$ containing $Q \cup R$, we say that $Q \cup R$ nests in $S$, and if $S$ is the smallest such subcircuit, we say that $S$ is the parent of $Q \cup R$, or that $Q \cup R$ is the daughter of $S$.

\begin{eg}Consider the Eulerian path defined \bas  \begin{matrix} A &B &C &C &D &B &D &E &  \end{matrix} \eas represented by the graph $G$ as in example \ref{eg:notpathrespecting}. 

Recall that two nested circuits are considered isolated, they are edge disjoint. Similarly, two disjoint circuits are isolated as they cannot share and edge. However, two overlapping circuits are not, as they are not edge disjoint by construction.

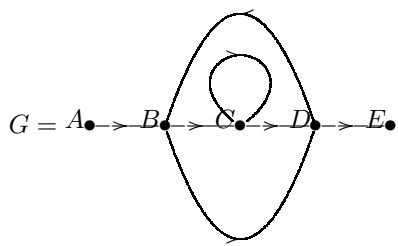
\bas  
G = {\xy
(-20, 0) *{\bullet} = "A" +(-2, 1)*{A},
(-10, 0) *{\bullet} = "B" +(-2, 1)*{B},
(0, 0) *{\bullet} = "C" +(-2, 1)*{C},
(10, 0) *{\bullet} = "D" +(-2, 1)*{D},
(20, 0) *{\bullet} = "E" +(-2, 1)*{E},
"A"; "B" **{\dir{-}}?/0pt/*{\dir{>}},
"B"; "C" **{\dir{-}}?/0pt/*{\dir{>}},
"C"; "D" **{\dir{-}}?/0pt/*{\dir{>}},
"D"; "E" **{\dir{-}}?/0pt/*{\dir{>}},
"D"; "B" **\crv{+(10,30)}?/0pt/*{\dir{>}},
"B"; "D" **\crv{+(-10,-30)}?/0pt/*{\dir{>}},
"C"; "C" **\crv{+(-7,5)&+(7,9)&+(7,-9)}?/0pt/*{\dir{>}},
\endxy}
\eas

There are several Eulerian subcircuits in this Eulerian path, namely: $BCCDB$, $CC$, and $DBD$. We write these using brackets in the full Eulerian path, \bas \begin{matrix} A &B &C &C &D &B &D &E &  \\ & [ & &  & &] & & &(1)\\ & & [ & ] & & & & &(2)\\ & & &  &[ & &] & &(3) \end{matrix} \;. \eas Note that circuit $(1)$ is not a cycle, while the circuits (2) and (3) are both cycles. 

Next, note that cycle (2) is nested inside circuit (1), circuit (1) and cycle (3) overlap, while (2) and (3) are disjoint.  In this example, since the sequence $(1) \setminus (2)$ is a cycle, we consider this an interrupted redundancy of the path. 

Therefore, in this example, we see that there are exactly 2 isolated loops in the path, either $(1) \setminus (2)$ and (2) or (2) and (3).

Note that while (1) and (3) are overlapping, the union $(1) \cup (3)$ is not a circuit. However, as noted above, we will consider the subpath $(1) \cup (3)$ as the daughter of $P$. 
\end{eg}

It is important to observe from the above example that counting isolated loops gives us a number of edge distinct cycles in the graph corresponding to an Eulerian path, not which cycles are isolated. In general, there are many possible sets of isolated loops in any given Eulerian path.

Unfortunately, depending on how one counts loops, there is not a unique number of isolated loops. Next, we give a useful definition,and then the  algorithm to count the minimum (resp. maximum) number of isolated loops in a path. 

\begin{dfn}\label{def:chi}
Given an Eulerian path $P$, let $\chi(P) = \begin{cases} 1 & \textrm{if } P \textrm{ is not a circuit} \\ 0 & \textrm{if } P  \textrm{ is a circuit}\end{cases}$.
\end{dfn}

\begin{algorithm} \label{alg:cara}
Given an Eulerian path $P$, represented by a graph $G$, one can find the minimum (resp. maximum) number of isolated loops in said path by the following algorithm
\begin{enumerate}
    \item For each subcycle $C$ in $P$ give it a weight $1$. Denote this $\omega(C)  = 1$.
    \begin{enumerate}
        \item If $P$ does not contain any subcycles, define $\omega(P) = 0$.
    \end{enumerate}
    \item Let $R = \{R_j\}$ be a set of disjoint weighted paths of $P$, with a common unweighted parent $Q \subsetneq P$. If the parent $Q$ is an interrupted redundancy (i.e. $Q \setminus R$ is an oriented cycle in $G$) or is all of $P$, give $Q$ the weight $\omega(Q) = \sum_{R_j \in R} \omega(R_j) + 1$.
    \item Let $S = \{S_j\}$ be a maximal set of overlapping redundancies such that no element is a parent of another. That is, every redundancy in $P$ is either in $S$ or disjoint from the redundancies in $S$.  
    \begin{enumerate}
        \item For each pair $S_i$, $S_j \in S$, that have been assigned a weight, if $S_i$ and $S_j$ are disjoint, remove the elements $S_i$ and $S_j$ from $S$ and replace with the union. \bas S = S \setminus \{S_i, S_j\} \cup (S_i \cup S_j)\;.\eas
        Assign the weight $\omega(S_i \cup S_i) = \omega(S_i) + \omega(S_j)$.
        \item For each overlapping pair $S_i$, $S_j \in S$ that have been assigned a weight, assume without lack of generality that $\omega(S_i) \leq \omega(S_j)$, drop the redundancy with the greater (resp. smaller) weight: \bas S = S \setminus S_j  \; \quad \; \textrm{(resp.  } S = S \setminus S_i \;\textrm{)}\eas
        \item If $S$ does not contain multiple elements, then the weight of $S$ is the weight of its single member.
    \end{enumerate}
    \item Assign to the path $P$ the weight \bas \omega(P) = \sum_{\substack{S \textbf{ maximal set of} \\ \textrm{overlapping redundancies }}}\omega(S) + \chi(P) \eas
\end{enumerate}
The weight of $P$ gives the maximal (resp. minimal) number of isolated loops in the path.
\end{algorithm}

Note that the difference in the algorithm for computing the maximal and minimal weights comes at the choice of weight associated the an overlapping set of circuits. Namely, every time there is a set of overlapping circuits, there is a choice to be made of which to keep and which to disregard. In order to calculate the maximal weight, consistently disregard the circuit with the smaller weight. In order to calculate the minimal weight, consistently disregard the circuit with the larger weight.

\begin{proof}
This algorithm gives a minimal (resp. maximal) weight to the Eulerian path $P$. It remains to check that every redundancy is considered in this algorithm, and that the algorithm returns the minimal (resp. maximal) count.

Let $G$ be the oriented multi-directed graph representing $P$. 

If $C$ is subcycle of $P$, then, by definition of a cycle, it does not contain any further redundancies. Therefore, the weight of $C$ is $1$.

The algorithm proceeds by induction on the poset structure on $P$ introduced by inclusion. 

Let $Q$ be an unweighted subcircuit of $P$ such that all of its daughters are weighted.  First, we consider the case when $R$ is the set of daughters of $Q$, and $Q \setminus R$ is an oriented circuit of $G$. We claim this implies that $Q\setminus R$ then is a cycle. If not, one could write $Q \setminus R = Q_1 \cup Q_2$ with $Q_1$ and $Q_2$ two subcircuits of $G$ and redundancies of $P$. Since $Q_1$ and $Q_2$ are comprised of adjacent edges of $G$, this implies that some daughter $R_i \in R$ must be either a subcircuit of $Q_1$ or $Q_2$. That is, $Q$ is not the parent of $R_i$. 

The algorithm gives the parent, $Q$, a weight of one more than the sum of all its disjoint daughters. That is, we interpret this situation as a parent redundancy being interrupted by many disjoint daughter circuits. Therefore, count all the redundancies in each of the disjoint daughters and add one for the parent, it the parent is also a redundancy.

%For $R$ a daughter of $Q$, if $Q \setminus R$ is not an oriented circuit of $G$, then $Q$ is the union of overlapping weighted circuits. Namely, it is the overlap that prevents $Q \setminus R$ from starting and ending at the same vertex. Since these are overlapping circuits, we do not wish to count both the redundancies in $R$ and those in $Q$. If we wish for the minimal (resp. maximal) count of isolated redundancies, we wish to consider the lower (resp. higher) weight of the redundancies. 

Given a set of overlapping redundancies, $S$, there may be multiple subset that do not overlap with each other. As above, replacing $P$ with $S$, these weights can be calculated independently, and then summed. In order to assign a weight to the subpath $S$, we preform a pairwise comparison of the union of these disjoint elements. Since taking the minimum (resp. maximum) is transitive, this is guaranteed to give the correct weight to the path.  
\end{proof}

We note that Algorithm \ref{alg:cara} only gives a minimal (resp. maximal) count for the number of isolated loops in any Eulerian path $P$. It does not indicate which loops are included in the count.

\begin{eg} \label{eg:nonunique}
Consider the path  \bas  P = \setcounter{MaxMatrixCols}{15} \begin{matrix} A & B & C & D & E & B & F & G & D & C & F& H & \\  & [ &  &  &  & ] &  &  &  &  &  &  & (1) \\  &  &  &  &  &  & [ &  &  &  & ] &  & (2) \\ &  & [ &  &  &  &  &  &  & ] &  &  & (3) \\  &  &  & [ &  &  &  &  & ] &  &  &  & (4) \\ \end{matrix} \;. \eas  and its graph representation \bas
G = {\xy
(-30, 0) *{\bullet} = "A" +(-2, -2)*{A},
(-10, 0) *{\bullet} = "B" +(-2, -2)*{B},
(0, 10) *{\bullet} = "C" +(-3, 1)*{C},
(0, 20) *{\bullet} = "D" +(-3, -2)*{D},
(-20, 15) *{\bullet} = "E" +(-2, 1)*{E},
(20, 15) *{\bullet} = "G" +(2, 1)*{G},
(10, 0) *{\bullet} = "F" +(2, -2)*{F},
(30, 0) *{\bullet} = "H" +(2, -2)*{H},
(-10, 25) = "X",
(10, 25) = "Y",
"A"; "B" **{\dir{-}}?/0pt/*{\dir{>}},
"B"; "C" **\crv{+(-3,-6)}?/0pt/*{\dir{>}},
"C"; "D" **\crv{+(3,-5)}?/0pt/*{\dir{>}},
"X"; "D" **\crv{+(-3, 5)},
"X"; "E" **\crv{+(1, 10)}?/-10pt/*{\dir{>}},
"E"; "B" **\crv{+(-8,4)}?/0pt/*{\dir{>}},
"B"; "F" **\crv{+(-10,-5)}?/0pt/*{\dir{>}},
"Y"; "D" **\crv{+(3, 5)},
"Y"; "G" **\crv{+(-1, 10)}?/-10pt/*{\dir{<}},
"D"; "C" **\crv{+(-3,5)}?/0pt/*{\dir{>}},
"F"; "C" **\crv{+(3,-6)}?/0pt/*{\dir{<}},
"G"; "F" **\crv{+(10,5)}?/0pt/*{\dir{<}},
"F"; "H" **{\dir{-}}?/0pt/*{\dir{>}}
\endxy} \;.
\eas  In this graph, the path $P$ is the parent of the cycles $(1)$ and $(2)$, the circuit $(3)$, and the  overlapping redundancies $(1) \cup (3)$, $(2) \cup (3)$ and $(1 \cup (2) \cup (3)$. The circuit $(3)$ is the parent of the cycle $(4)$. Note that $(1)$, $(2)$ and $(4)$ are cycles in this path, but $(3)$ is not. The algorithm assigns a weight of one to the cycles $(1)$, $(2)$ and $(4)$. 

Since $P$ is the parent of $(1)$ and $(2)$, there is only one choice for $R$: $R = \{(4)\}$. The cycle $(4)$ has an unweighted parent $(3)$ which is given weight $2$. Since the parent of $(4)$ is $P$, we stop this step.

Furthermore, there is only one choice for $S$: $S = \{(1), (2), (3)\}$. Since $(1)$ and $(2)$ are disjoint, they are replaced with the path $(1) \cup (2)$, and which has weight $\omega ((1) \cup (2)) = \omega (1) + \omega (2) = 2$. Now redefine $S = \{ (1) \cup (2), (3)\}$. Since both elements of $S$ have the same weight, it does not matter which is removed by the algorithm, only that one is. Then $S$ has weight $2$.

Since $P$ is a parent of $S$, and $P$ is not a circuit, $\omega(P) = 2$. Note that Algorithm \ref{alg:cara} does not specify which are the isolated loops in $P$. Either the set $\{(1) = BCDEB, (2) = FGDCF \}$ or $\{(3) = DEBFGD, (4) = CDEBFGDC \}$ qualify. 

\end{eg}

It is useful to give an easily calculable upper bound to the maximal number isolated loops for an Eulerian path given by Algorithm \ref{alg:cara}.

\begin{thm}\label{res:upperbound}
Consider an Eulerian path, $P$, represented by a multi-digraph $G$. Let $V_{>1}(G)$ be the set of vertices of $G$ with degree greater than $1$. Denote by $\deg(v)$ the degree of the vertex $v$, and $\omega_{\max}(P)$ the maximal number of isolated loops in $P$. Then \bas \omega_{\max}(P) \leq \sum_{v \in V_{>1}(G)} (\frac{\deg(v)}{2} - 1) + \chi(P) \;.\eas 
\end{thm}

\begin{proof}
We claim that $\sum_{v \in V_{>1}} (\frac{\deg(v)}{2} -1)$ give the total number of redundancies of a path $P$, not just the isolated loops. Therefore, it is an upper bound.

Indeed, since $P$ is an Eulerian path, every vertex, except possibly the start and end point of $P$ has even degree. If $P$ is an Eulerian circuit, then its the initial vertex also has even degree. For even vertices with even degree, $\frac{\deg(v)}{2}$ counts the number of times a path visits that vertex. In other words, $\frac{\deg(v)}{2} - 1$ counts the number of circuits starting from $v$. Summing this gives give the total number of subcircuits in the path. If $P$ is not an Eulerian circuit, then the first and last vertices each contribute $-\frac{1}{2}$ to this sum, which is accounted for by the indicator function $\chi(P)$. 
\end{proof}

\section{Comparison to other measures\label{sec:comparison}}

In this section, we discuss several measures of graph redundancy commonly used in the literature, and compare them to the measure developed in this paper. To see the differences between the different redundancy measures described in this section, we refer back to the following example. 

\begin{eg}\label{eg:twoloopgraph}
As a running example in this section, consider the Eulerian path and graph as in Example \ref{eg:nonunique}: \bas A \; B \; C \; D\;E\;B\;F\;G\;D\;C\;F\;H \;,\eas represented by 

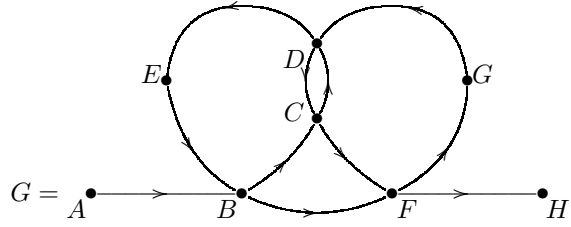
\bas  
G = {\xy
(-30, 0) *{\bullet} = "A" +(-2, -2)*{A},
(-10, 0) *{\bullet} = "B" +(-2, -2)*{B},
(0, 10) *{\bullet} = "C" +(-3, 1)*{C},
(0, 20) *{\bullet} = "D" +(-3, -2)*{D},
(-20, 15) *{\bullet} = "E" +(-2, 1)*{E},
(20, 15) *{\bullet} = "G" +(2, 1)*{G},
(10, 0) *{\bullet} = "F" +(2, -2)*{F},
(30, 0) *{\bullet} = "H" +(2, -2)*{H},
(-10, 25) = "X",
(10, 25) = "Y",
"A"; "B" **{\dir{-}}?/0pt/*{\dir{>}},
"B"; "C" **\crv{+(-3,-6)}?/0pt/*{\dir{>}},
"C"; "D" **\crv{+(3,-5)}?/0pt/*{\dir{>}},
"X"; "D" **\crv{+(-3, 5)},
"X"; "E" **\crv{+(1, 10)}?/-10pt/*{\dir{>}},
"E"; "B" **\crv{+(-8,4)}?/0pt/*{\dir{>}},
"B"; "F" **\crv{+(-10,-5)}?/0pt/*{\dir{>}},
"Y"; "D" **\crv{+(3, 5)},
"Y"; "G" **\crv{+(-1, 10)}?/-10pt/*{\dir{<}},
"D"; "C" **\crv{+(-3,5)}?/0pt/*{\dir{>}},
"F"; "C" **\crv{+(3,-6)}?/0pt/*{\dir{<}},
"G"; "F" **\crv{+(10,5)}?/0pt/*{\dir{<}},
"F"; "H" **{\dir{-}}?/0pt/*{\dir{>}}
\endxy}
\eas

\end{eg}
 
\subsection{Euler's loop number}
Euler’s loop number is used to count the number of independent loops in a given \emph{planar graph}, i.e. a graph that can be drawn on a piece of paper such that none of its edges cross. The number of independent loops, $L$, of such a graph is given by Euler's formula $L=E-V+C$, where $E$ is the number of edges in a graph, $V$ is the number of vertices, and $C$ the number of connected components of the graph. For more on planar graphs and Euler's loop number, see \cite[Chapter 5]{Graphtheorybook}. Note that in this paper, as we only consider connected graphs, $C = 1$. This quantity is closely related to the Euler characteristic of a graph. Namely, the quantity $L + 1$ counts the number of 2 dimensional spaces, or \emph{faces}, the graph divides the plane into, including the infinite space on the outside of the graph. This quantity is constant no matter how the graph is drawn in the plane, provided the edges remain non-intersecting. Therefore, the quantity $L$ gives the number of finite faces defined by the planar graph $G$, or independent loops.

Note that while each finite face defined by the planar graph $G$ corresponds to a cycle within $G$, not all cycles are counted in this manner. In particular, any simple cycle that is formed by the union of these finite faces is not included in this map. For this reason, one often says that $L$ gives the number of \emph{independent} loops in a planar graph. Note that unlike the definition of isolated loops, the independent loops may have overlapping edges amongst the loops. 

There are two main ways in which the quantity $L$ differs from the isolated loop count given in Algorithm \ref{alg:cara}. First, Euler's loop number does not require an orientation on the edges, which the Isolated loop count, built off an Eulerian path, does. Therefore, in the graph in Example \ref{eg:twoloopgraph}, the number of independent loops of $L$ is $4$. From the embedding given in the example, these are $BCDEB$, $CDC$, $FGDCG$ and $BFCB$. First note that a different embedding would give rise to a different set of faces, and thus a different set of simple cycles. Also, note that this last cycle listed above is not one of the oriented cycles of the graph $G$.

The other difference between Euler's loop number and the Isolated loop count is that the former is only defined on planar graphs, while there is no reason than an Eulerian path needs to be planar (for example, consider $K_5$.\footnote{It can be shown (for instance \cite[Corollary 13.4]{Graphtheorybook}) that if $K_5$ were a planar graph, its edges and vertices would satisfy $E \geq 3V - 6$. However, $K_5$ has $10$ edges and $5$ vertices, so this relation does not hold, and thus the graph cannot be planar.}  Specifically, the minimum isolated loop count of $K_5$ is 1, while the maximum is $2$, while the number of independent loops is not defined. We do, however, note that there is a generalization of the Euler loop number, performed by relating the Euler characteristic of a graph to the genus of the space it is embedded in. However, such analysis is beyond the scope of this paper.

\subsection{Counting oriented loops}

All the graphs considered in this paper arise from Eulerian paths, and therefore are multi-digraphs. There are many algorithms that exist that count the number of oriented loops in such a graph \cite{countingcycles, findingcycles}. However, these algorithms do not require the underlying graphs to be Eulerian. That is, given an Eulerian path, these algorithms count the directed cycles without regard to whether or not the edges of that cycle respect the order imposed on the graph by said path. For instance, the cycle $CDC$ in the graph in Example \ref{eg:nonunique} is a valid simple oriented cycle. However, this oriented cycle does not appear in the Eulerian path that we wish to study, though it may appear in other paths that also correspond to this graph. 

Furthermore, the algorithms to count oriented cycles of a graph count \emph{all} simple oriented cycles of a graph, regardless of whether or not they share edges with each other. Therefore, both the simple loops $BCDEB$ and $DEBFGD$ in the graph in Example \ref{eg:nonunique} are counted by algorithms counting oriented loops while Algorithm \ref{alg:cara} counts at most one of them. 

Finally, it is worth noting that unlike for Euler loop number, the algorithms for counting oriented loops do not require the graphs to be planar. Furthermore these algorithms do not have a concept of independence of loops. Rather they count all oriented cycles of the graph. 

\subsection{In Degree and Out degree}

When studying Eulerian paths, we know that every vertex, barring possibly the first and last, has even degree, as every vertex, except the first and the last has the same number of edges entering it as leaving it: \bas \frac{\deg(v)}{2} = \textrm{in degree} = \textrm{out degree} \;.\eas 

We may use this property of Eulerian paths to give an upper bound for the isolated loop count:  \bas (\frac{\deg(v)}{2} - 1) + \chi(P) \eas as  given in Theorem \ref{res:upperbound}). Furthermore, considering the in degree of an Eulerian graph gives a lower bound on the number of oriented cycles of that graph. First, because we consider Eulerian graphs that need not be Eulerian circuits, we need to generalize the concept of in and out degrees: \begin{enumerate}
    \item If $G$ is an Eulerian circuit, then $\frac{\deg(v)}{2}$ counts the number of times that a vertex is visited.
    \item If $G$ is not a Eulerian circuit but still an Eulerian path, then $\frac{\deg(v)}{2}$ counts the number of times the path visits a vertex for all vertices but the first and last, while $\frac{\deg(v)}{2} + \frac{1}{2}$ counts the number of times the path encounters the first and last vertices. 
\end{enumerate}
In other words, let $G$ be a multi-digraph with edges oriented to respect an Eulerian path on it. For any vertex $G$, $\lfloor \frac{\deg(v)}{2} \rfloor$ gives the number of Eulerian circuits starting at that vertex. Therefore, the expression $\sum_{v \in V(G)} (\frac{\deg(v)}{2} - 1) + \chi(P)$ from Theorem \ref{res:upperbound} counts the number of times an Eulerian path returns to a vertex $v$. As such, it is a lower bound for the number of oriented cycles in an Eulerian graph. 

For a graph that cannot represent an Eulerian path, counting the in or out degree of a vertex does not yield any useful information.

\subsection{Counting Eulerian circuits \label{sec:BEST}}

There is a formula for counting the number of Eulerian circuits in an Eulerian graph \cite{BEST}. In our more general case, if we are considering a graph, as in Example \ref{eg:twoloopgraph} that is defined by an Eulerian path with distinct endpoints, we may turn this into an Eulerian circuit by adding a vertex $X$ and two edges, one from $X$ to the start of the path and one from the final vertex to $X$. The resulting multi-digraph is an Eulerian circuit. 

Given an Eulerian graph, there may be many Eulerian circuits on it. The de Bruijn, van Aardenne-Ehrenfest, Smith and Tutte (BEST) theorem \cite{BEST} counts the number of distinct Eulerian circuits on a given Eulerian graph. This is can be expressed in terms of the in degrees of all the vertices of the graph given by the formula \bas ec(G) = t(G) \prod_{v \in V(G)}(\frac{\deg(v)}{2}-1)! \;,\eas where $t(G)$ is the number of oriented spanning trees flowing towards a fixed vertex in $G$, and $ec(G)$ represents the number of Eulerian circuits. This number is unique for Eulerian graphs, and can be computed by Kirchhoff's theorem for directed multi-graphs. In the running example we have three different Eulerian circuits that are compatible with the edge orientation of $G$. The edges of the corresponding trees are indicated by the $\textrm{--}$ in the ordered sequence of vertices below. \bas X \; A \textrm{--} B \; C \; D\;E\textrm{--}B\textrm{--}F\;G\textrm{--}D\textrm{--}C\textrm{--}F\textrm{--}H\textrm{--} X \\ X\; A \textrm{--} B \; C \; D\;C\textrm{--}F\;G\textrm{--}D\textrm{--}E\textrm{--}B\textrm{--}F\textrm{--}H\textrm{--}X \\ X \; A \textrm{--} B \; F \; G\textrm{--}D\;E\textrm{--}B\;C\;D\textrm{--}C\textrm{--}F\textrm{--}H \textrm{--}X \eas 

First note that the BEST theorem gives the number of distinct Eulerian circuits that can traversed in a given graph. Unlike any of the other methods discussed in this section, this theorem does not return a number of simple oriented cycles. 

Finally, note that in order to perform this calculation, one had to consider the graph $G$ as an multi-digraph, not as an Eulerian path, i.e. without the ordering of the edges of the graph. This is very different from the situation in Algorithm \ref{alg:cara}, where the count depends on the order in which the edges are encountered.

\end{document}